\documentclass[a4paper]{article}

\usepackage[english]{babel}
\usepackage[utf8]{inputenc}
\usepackage{amsmath}
\usepackage{graphicx}
\usepackage[colorinlistoftodos]{todonotes}
\usepackage[a4paper,left=2cm,right=2cm,top=2.cm,bottom=2cm]{geometry}
\usepackage{authblk}
\usepackage{txfonts}
\usepackage{bbold} 
\usepackage{amsopn} 
\usepackage{graphicx}
\newtheorem{example}{Example}
\newtheorem{proposition}{Proposition}
\newtheorem{definition}{Definition}
\newtheorem{lemma}{Lemma}
\newenvironment{proof}{\noindent\textit{Proof~~}}
{\nolinebreak[4]\hfill$\blacksquare$\\\par}
\title{Strong isomorphism in Eisert-Wilkens-Lewenstein type quantum games}

\author{Piotr Fr\c{a}ckiewicz}
\affil{Institute of Mathematics\\ Pomeranian University, Poland}
\date{\today}

\begin{document}
\maketitle

\begin{abstract}
The aim of this paper is to bring together the notions of quantum game and game isomorphism. The work is intended as an attempt to introduce a new criterion for quantum game schemes. The generally accepted requirement forces a quantum scheme to generate the classical game in a particular case. Now, given a quantum game scheme and two isomorphic classical games, we additionally require the resulting quantum games to be isomorphic as well. We are concerned with the Eisert-Wilkens-Lewenstein quantum game scheme and the strong isomorphism between games in strategic form.
\end{abstract}
\section{Introduction}
Sixteen years of research on quantum games have given us many ideas of how quantum games could be described. For example, we have learned from \cite{mejer} that players who are allowed to use some specific unitary operators may gain an advantage over the players who use only classical strategies. The schemes introduced in \cite{ewl} and \cite{marinatto} give us two different ways of describing quantum $2\times 2$ games. The paper \cite{du}, in turn, provides us with a quantum scheme for the Cournot duopoly game. What connects these protocols is the capability to obtain the classical game. This appears to be a generally accepted necessary condition imposed on a quantum scheme. One can also find other and more subjective guidelines for quantum game schemes. Paper \cite{fracormwcomment} shows how to generalize the scheme introduced in \cite{marinatto} by assuming that the new model should output the classical game (up to the order of players' strategies) if the initial state is one of the computational basis states. In addition, the work of Bleiler \cite{bleiler} distinguishes between {\it proper} and more strict {\it complete} quantization. Roughly speaking, the first notion concerns quantum schemes where the counterparts of classical pure strategies can be found in pure quantum strategies. The second one requires the quantum strategy set to include the counterparts of the mixed classical strategies. With these notions the Marinatto-Weber (MW) \cite{marinatto} scheme turns out to be not even a proper quantization. The Eisert-Wilkens-Lewenstein (EWL) \cite{ewl} scheme, in turn, is a complete quantization, and this is the case as long as the players' quantum strategies include the one-parameter unitary operators $U(\theta,0,0)$ (see formula (\ref{generalunitary})). In particular, we can find a lot of papers where the EWL scheme was studied with the two-parameter unitary strategies \cite{experimental}, \cite{chen}, \cite{liiqbal}, \cite{nawaznowy}. However, as it was noted in \cite{benjamincomment} the set $\{U(\theta, \alpha,0)\}$ appears not to reflect any reasonable
physical constraint as this set is not closed under composition. Moreover, \cite{flitney} showed that different two-parameter strategy spaces in the EWL scheme imply different sets of Nash equilibria. In this paper we explain why the set of two-parameter unitary operators may not be reasonable from the game theory viewpoint. Our criterion is formulated in terms of isomorphic games. If we assume that both classical games are the same with respect to game theoretical tools, we require the corresponding quantum games to be equivalent in the same way. It is worth noting that quantum game schemes introduced in \cite{ewl}, \cite{marinatto} and the refined MW scheme defined in \cite{fracor2015} preserve the so-called {\it strategic equivalence}. We recall what this means. The following definition can be found in \cite{maschler}. See the Preliminary section for the definition of strategic form game $(N, (S_{i})_{i\in N}, (u_{i})_{i\in N})$ and its components.  
\begin{definition}
Two games in strategic form $(N, (S_{i})_{i\in N}, (u_{i})_{i\in N})$ and $(N, (S_{i})_{i\in N}, (v_{i})_{i\in N})$ with the same set of players and the same sets of pure strategies are strategically equivalent if for each player $i \in N$ the function $v_{i}$ is a positive affine transformation of the function $u_{i}$. In other words, there exist $\alpha_{i} > 0$ and $\beta_{i} \in \mathbb{R}$ such that
\begin{equation}\label{strategiceq}
v_{i}(s) = \alpha_{i}u_{i}(s) + \beta_{i},~~\mbox{for each}~~s\in \prod_{i\in N}S_{i}.
\end{equation}
\end{definition}
It is clear that given two strategically equivalent games $\Gamma$ and $\Gamma'$, the player $i$'s payoff operators $M_{i}$ and $M'_{i}$ in the quantum games are connected by equation $M'_{i} = \alpha_{i}M_{i} + \beta_{i}$. Then, by linearity of trace, the quantum payoff functions satisfy~(\ref{strategiceq}), i.e., $\mathrm{tr}(\rho_{\mathrm{fin}}M'_{i}) = \alpha_{i}\mathrm{tr}(\rho_{\mathrm{fin}}M_{i}) + \beta_{i}$. Virtually, strategically equivalent games $\Gamma$ and $\Gamma'$ describe the same game-theoretical problem. In particular, every equilibrium (pure or mixed) of the game $\Gamma$ is an equilibrium of the game $\Gamma'$.

The strategy equivalence can be extended to take into account different orders of players' strategies. This type of equivalence is included in the definition of strong isomorphism. Clearly, if for example, two bimatrix games differ only in the order of a player's strategies we still have the games that describe the same problem from the game theoretical viewpoint. Given a quantum scheme, it appears reasonable to assume that the resulting quantum game will not depend on the numbering of players' strategies in the classical game. As a result, if there is a strong isomorphism between games, we require that the quantum counterparts of these games are also isomorphic.
\section{Preliminaries}
In order to make our paper self-contained we give the important preliminaries from game theory and quantum game theory.
\subsection{Strong isomorphism}
First we recall the definition of strategic form game \cite{maschler}.
\begin{definition}
A game in strategic form is a triple $\Gamma = (N, (S_{i})_{i\in N}, (u_{i})_{i\in N})$ in which
\begin{itemize}
\item $N = \{1,2,\dots,n\}$ is a finite set of players.
\item $S_{i}$ is the set of strategies of player $i$, for each player $i\in N$.
\item $u_{i}\colon S_{1}\times S_{2} \times \dots \times S_{n} \to \mathbb{R}$ is a function associating each vector of strategies $s=(s_{i})_{i\in N}$ with the payoff $u_{i}(s)$ to player $i$, for every player $i\in N$.
\end{itemize}
\end{definition}
The notion of strong isomorphism defines classes of games that are the same up to numbering of the players and the order of players' strategies.
The following definitions are taken from \cite{garcia} (see also \cite{nash}, \cite{peleg1} and \cite{peleg2}). The first one defines a mapping that associates players and their actions in one game with players and their actions in the other game.
\begin{definition}
Given $\Gamma = (N, (S_{i})_{i\in N}, (u_{i})_{i\in N})$ and $\Gamma' = (N, (S'_{i})_{i\in N}, (u'_{i})_{i\in N})$, a game mapping $f$ from $\Gamma$ to $\Gamma'$ is a tuple $f = (\eta, (\varphi_{i})_{i\in N})$ where $\eta$ is a bijection from $N$ to $N$ and for any $i\in N$, $\varphi_{i}$~is a bijection from $S_{i}$ to $S_{\eta(i)}$.
\end{definition} 
\begin{example}\label{ex1}
\textup{Let us consider two bimatrix games}
\begin{equation}\label{gamesex1}
\bordermatrix{& l & r \cr t & (a_{00}, b_{00}) & (a_{01}, b_{01}) \cr b & (a_{10}, b_{10}) & (a_{11}, b_{11})} \quad \mbox{\textup{and}} \quad \bordermatrix{& l' & r' \cr t' & (a'_{00}, b'_{00}) & (a'_{01}, b'_{01}) \cr b' & (a'_{10}, b'_{10}) & (a'_{11}, b'_{11})}.
\end{equation}
\textup{Then, $N = \{1,2\}$ and $S_{1} = \{t,b\}$, $S_{2} = \{l,r\}$, $S'_{1} = \{t',b'\}$, $S'_{2} = \{l',r'\}$. As an example of a game mapping let $f = (\eta, \varphi_{1}, \varphi_{2})$,} 
\begin{equation}\label{exmapping}
\eta= (1\to 2, 2 \to 1),\; \varphi_{1} = (t\to l', b \to r'),\; \varphi_{2} = (l\to b', r\to t').
\end{equation}
\textup{Since $\varphi_{1} \colon S_{1} \to S'_{2}$ and $\varphi_{2} \colon S_{2} \to S'_{1}$, it follows that $f$ maps $(s_{1},s_{2}) \in S_{1} \times S_{2}$ to $(\varphi_{2}(s_{2}), \varphi_{1}(s_{1}))$. From~(\ref{exmapping}) we conclude that}
\begin{equation}
f = ((t,l) \to (b', l'), (t,r) \to (t', l'), (b,l) \to (b', r'), (b,r) \to (t', r')).
\end{equation}
\end{example}
In general case, mapping $f$ from $(N, (S_{i})_{i\in N}, (u_{i})_{i\in N})$ to $(N, (S'_{i})_{i\in N}, (u'_{i})_{i\in N})$ identifies player $i\in N$ with player $\eta(i)$ and maps $S_{i}$ to $S_{\eta(i)}$. This means that strategy profile $(s_{1},\dots, s_{n}) \in S_{1} \times \dots \times S_{n}$ is mapped into profile $(s'_{1}, \dots, s'_{n})$ that satisfies equation $s'_{\eta(i)} = \varphi_{i}(s_{i})$ for $i\in N$.

The notion of game mapping is a basis for definition of game isomorphism. Depending on how rich structure of the game is to be preserved we can distinguish various types of game isomorphism. One that preserves the players' payoff functions is called the strong isomorphism. The formal definition is as follows:
\begin{definition}
Given two strategic games $\Gamma = (N,(S_{i})_{i\in N}, (u_{i})_{i\in N})$ and $\Gamma' = (N, (S'_{i})_{i\in N}, (u'_{i})_{i\in N})$, a game mapping $f = (\eta, (\varphi_{i})_{i\in N})$ is called a strong isomorphism if relation $u_{i}(s) = u'_{\eta(i)}(f(s))$ holds for each $i\in N$ and each strategy profile $s\in S_{1}\times \dots \times S_{n}$. 
\end{definition} 
From the above definition it may be concluded that if there is a strong isomorphism between games $\Gamma$ and $\Gamma'$, they may differ merely by the numbering of players and the order of their strategies.
\begin{example}
\textup{Let $f$ be a game mapping defined in Example~\ref{ex1}. By definition, $f$ becomes the strong isomorphism if condition $u_{i}(s) = u'_{\eta(i)}(f(s))$ is imposed on the payoffs in~(\ref{gamesex1}). This gives}
\begin{align}\label{twoequations}
\begin{split}
a_{00} = b'_{10},\; a_{01} = b'_{00},\; a_{10} = b'_{11},\; a_{11} = b'_{01}, \cr
b_{00} = a'_{10},\; b_{01} = a'_{00},\; b_{10} = a'_{11},\; b_{11} = a'_{01}, \end{split}
\end{align}
\textup{where, for instance, $a_{01} = b'_{00}$ follows from equation $u_{1}((t,r)) = u'_{2}((t',l'))$. Substituting~(\ref{twoequations}) into~(\ref{gamesex1}) we conclude that games}
\begin{equation}\label{gamesex2}
\bordermatrix{& l & r \cr t & (a_{00}, b_{00}) & (a_{01}, b_{01}) \cr b & (a_{10}, b_{10}) & (a_{11}, b_{11})} \quad \mbox{\textup{and}} \quad \bordermatrix{& l' & r' \cr t' & (b_{01}, a_{01}) & (b_{11}, a_{11}) \cr b' & (b_{00}, a_{00}) & (b_{10}, a_{10})}.
\end{equation}
\textup{are isomorphic. In this case, the games differ by the numbering of players and the order of strategies of player 2. Indeed, in the second game of~(\ref{gamesex2}) player 1 and 2 choose now between columns and rows, respectively. Moreover, player 1's first (second) strategy still guarantees the payoff $a_{00}$ or $a_{01}$ ($a_{10}$ or $a_{11}$) whereas player 2's strategies are interchanged: the first one implies now the payoff $b_{01}$ or $b_{11}$. }
\end{example}
Relabeling players or their strategies does not affect a game with regard to Nash equilibria. If $f$ is a strong isomorphism between games $\Gamma$ and $\Gamma'$, one may expect that the Nash equilibria in $\Gamma$ map to ones in $\Gamma'$ under $f$. We will prove the following lemma as it is needed throughout the paper.
\begin{lemma}\label{lemma1}
Let $f$ be a strong isomorphism between games $\Gamma$ and $\Gamma'$. Strategy profile $s^{*} = (s^*_{1}, \dots, s^*_{n}) \in S_{1} \times \dots \times S_{n}$ is a Nash equilibrium in game $\Gamma$ if and only if $f(s^*) \in S'_{1}\times \dots \times S'_{n}$ is a Nash equilibrium in $\Gamma'$. 
\end{lemma}
\begin{proof}
The proof is based on the following observation. Since $f((s_{1},\dots, s_{n})) = (s'_{1}, \dots, s'_{n})$ where $s'_{\eta(i)} = \varphi_{i}(s_{i})$, it follows that $f((s_{i}, s_{-i}))$ may be written as $\left(s'_{\eta(i)}, s'_{-\eta(i)}\right)$. As $f$ is an isomorphism, we have $u_{i}(s) = u'_{\eta(i)}(f(s))$ for each strategy profile $s$. Thus
\begin{equation}
u_{i}(s^*) = u'_{\eta(i)}\left(s^{*'}_{1}, \dots, s^{*'}_{n}\right)
\end{equation}
and
\begin{equation}
u_{i}\left(s_{i}, s^*_{-i}\right) = u'_{\eta(i)}\left(f\left(s_{i}, s^*_{-i}\right)\right) = u'_{\eta(i)}\left(s'_{\eta(i)}, s^{*'}_{-\eta(i)}\right).
\end{equation}
This allows us to conclude that the inequality
\begin{equation}
u_{i}(s^*) \geq u_{i}(s_{i}, s^*_{-i})
\end{equation}
holds for each $i\in N$ and each strategy $s_{i} \in S_{i}$ if and only if 
\begin{equation}
u'_{\eta(i)}\left(s^{*'}_{1}, \dots, s^{*'}_{n}\right) \geq u'_{\eta(i)}\left(s'_{\pi(i)}, s^{*'}_{-\eta(i)}\right)
\end{equation}
for each $\eta(i) \in N$ and each strategy $s'_{i} \in S'_{i}$.
This finishes the proof.
\end{proof}
\subsection{Eisert-Wilkens-Lewenstein scheme}
Let us consider a strategic game $\Gamma = (N, (S_{i})_{i\in N}, (u_{i})_{i\in N})$ with $S_{i} = \left\{s^i_{0}, s^i_{1}\right\}$ for each $i\in N$. The generalized Eisert-Wilkens-Lewenstein approach to game $\Gamma$ is defined by triple $\Gamma_{EWL} = \left(N, (D_{i})_{i\in N}, (M_{i})_{i\in N}\right)$, where
\begin{itemize}
\item $D_{i}$ is a set of unitary operators from $\mathsf{SU}(2)$. The commonly used parametrization for $U \in \mathsf{SU}(2)$ is given by
\begin{equation}\label{generalunitary}
U(\theta, \alpha, \beta) = \left(\begin{array}{ccc} \mathrm{e}^{\mathrm{i} \alpha}\cos{\frac{\theta}{2}} & \mathrm{i} \mathrm{e}^{\mathrm{i} \beta}\sin{\frac{\theta}{2}} \\ \mathrm{i} \mathrm{e}^{-\mathrm{i} \beta}\sin{\frac{\theta}{2}} & \mathrm{e}^{-\mathrm{i} \alpha}\cos{\frac{\theta}{2}}\end{array}\right), \; \theta \in [0,\pi], \; \alpha, \beta \in [0,2\pi).
\end{equation}
Then $D_{i}$ is assumed to include set $\{U(\theta, 0,0)\colon \theta\in [0,\pi]\}$. Elements $U_{i} \in D_{i}$ play the role of player $i$'s strategies. The players, by choosing $U_{i} \in D_{i}$, determine the final state $|\Psi\rangle$ according to the following formula:
\begin{equation}\label{finalstate}
|\Psi\rangle = J^{\dag}\left(\bigotimes^{n}_{i=1}U_{i}(\theta_{i}, \alpha_{i}, \beta_{i})\right)J|0\rangle^{\otimes n}\; \mbox{where}\; J = \frac{1}{\sqrt{2}}\left(\mathbb{1}^{\otimes n} + \mathrm{i}\sigma^{\otimes n}_{x}\right)
\end{equation}
($\mathbb{1}$ is the identity matrix of size 2 and $\sigma_{x}$ is the Pauli matrix $X$).
\item $M_{i}$ is an observable defined by the formula
\begin{equation}\label{payoffoperators}
M_{i} = \sum_{j_{1},\dots, j_{n} \in \{0,1\}}a^i_{j_{1}\dots j_{n}}|j_{1}\dots j_{n}\rangle \langle j_{1} \dots j_{n}|.
\end{equation}
The numbers $a^i_{j_{1}\dots j_{n}}$ are the player $i$'s payoffs in $\Gamma$ such that $a^i_{j_{1}\dots j_{n}} = u_{i}\left(s^i_{j_{1}}, \dots, s^i_{j_{n}}\right)$. The player $i$'s payoff $u_{i}$ in $\Gamma_{EWL}$ is defined as the average value of measurement $M_{i}$, i.e.,
\begin{equation}\label{payofffunctionq}
u_{i}\left(\bigotimes^n_{i=1}U_{i}(\theta_{i}, \alpha_{i}, \beta_{i})\right)\coloneqq \langle \Psi|M_{i}|\Psi\rangle.
\end{equation}
\end{itemize}
\section{Strong isomorphism in Eisert-Wilkens-Lewenstein quantum games}
Having specified the notion of strong isomorphism and the generalized Eisert-Wilkens-Lewenstein scheme we will now check if the isomorphism between the classically played games makes the corresponding quantum games isomorphic. We first examine the case when the players' unitary strategies depend on two parameters. The quantum game $\Gamma_{EWL}$ with 
\begin{equation}\label{twoparameter}
D_{i} = \{U_{i}(\theta_{i}, \alpha_{i}, 0)\colon \theta_{i} \in [0,\pi], \alpha_{i} \in [0,2\pi)\} 
\end{equation}
is particularly interested. That setting was used to introduce the EWL scheme \cite{ewl} and has been widely studied in recent years (see, for example, \cite{experimental}, \cite{chen}, \cite{liiqbal}, \cite{nawaznowy}). We begin with an example of isomorphic games that describe the Prisoner's Dilemma. 
\begin{example}
\textup{The generalized Prisoner's Dilemma game and one of its isomorphic counterparts may be given by the following bimatrices:}
\begin{equation}\label{p3}
\Gamma\colon ~~\bordermatrix{& l & r \cr t & (R, R) & (S, T) \cr b & (T, S) & (P, P)} \quad \mbox{\textup{and}} \quad \Gamma'\colon ~~ \bordermatrix{& l' & r' \cr t' & (S, T) & (R, R) \cr b' & (P, P) & (T, S)},
\end{equation}
\textup{where $T>R>P>S$. Note that the games are the same up to the order of player 2' strategies. Let us now examine the EWL approach to $\Gamma$ and $\Gamma'$ defined by triples}
\begin{align}\label{p2}
\begin{split}
&\Gamma_{EWL} = \left(\{1,2\}, (\{U_{i}(\theta_{i}, \alpha_{i}, 0\}\colon )_{i\in \{1,2\}}, (M_{i})_{i\in \{1,2\}}\right), \cr 
&\Gamma'_{EWL} = \left(\{1,2\}, (\{U'_{i}(\theta'_{i}, \alpha'_{i}, 0\})_{i\in \{1,2\}}, (M'_{i})_{i\in \{1,2\}}\right),
\end{split}
\end{align}
\textup{where}
\begin{align}\label{payoffoperators1}
\begin{split}
&(M_{1}, M_{2}) = (R,R)|00\rangle \langle 00| + (S,T)|01\rangle \langle 01| + (T,S)|10\rangle \langle 10| + (P,P)|11\rangle \langle 11|, \cr
&(M'_{1}, M'_{2}) = (S,T)|00\rangle \langle 00| + (R,R)|01\rangle \langle 01| + (P,P)|10\rangle \langle 10| + (T,S)|11\rangle \langle 11|.
\end{split}
\end{align}
\textup{We first compare the sets of Nash equilibria in $\Gamma_{EWL}$ and $\Gamma'_{EWL}$ to check if the games may be isomorphic. We recall from \cite{ewl} that there is the unique Nash equilibrium $U_{1}(0,\pi/2,0)\otimes U_{2}(0,\pi/2,0)$ in $\Gamma_{EWL}$ that determines the payoff profile $(R,R)$. When it comes to $\Gamma'_{EWL}$, we set $n=2$ in (\ref{finalstate}) and replace~(\ref{payoffoperators}) by $M'_{1}$ and $M'_{2}$ from~(\ref{payoffoperators1}). Then we can rewrite~(\ref{payofffunctionq}) as}
\begin{align}\label{p1}
&(u'_{1}, u'_{2})\left(U_{1}(\theta'_{1}, \alpha'_{1},0) \otimes U_{2}(\theta'_{2}, \alpha'_{2},0)\right)\nonumber\\  &\quad= (S,T)\left(\cos{(\alpha'_{1} + \alpha'_{2})\cos{\frac{\theta'_{1}}{2}}}\cos{\frac{\theta'_{2}}{2}} \right)^2 \nonumber \\ &\quad + (R,R)\left(\cos{\alpha'_{1}}\cos{\frac{\theta'_{1}}{2}}\sin{\frac{\theta'_{2}}{2}} + \sin{\alpha'_{2}}\sin{\frac{\theta'_{1}}{2}}\cos{\frac{\theta'_{2}}{2}} \right)^2 \nonumber \\ &\quad + (P,P)\left(\sin{\alpha'_{1}}\cos{\frac{\theta'_{1}}{2}}\sin{\frac{\theta'_{2}}{2}} + \cos{\alpha'_{2}}\sin{\frac{\theta'_{1}}{2}}\cos{\frac{\theta'_{2}}{2}} \right)^2 \nonumber \\ &\quad + (T,S)\left(\sin(\alpha'_{1} + \alpha'_{2})\cos{\frac{\theta'_{1}}{2}}\cos{\frac{\theta'_{2}}{2}} - \sin{\frac{\theta'_{1}}{2}}\sin{\frac{\theta'_{2}}{2}} \right)^2.
\end{align}
\textup{Let $U'_{2}(\theta'_{2}, \alpha'_{2}, 0)$ be an arbitrary but fixed strategy of player 2. Then it follows from~(\ref{p1}) that strategy $U'_{1}(\theta'_{1}, \alpha'_{1}, 0)$ specified by equation}
\begin{equation}\label{equation2cases}
U'_{1}(\theta'_{1}, \alpha'_{1}, 0) = \begin{cases} U'_{1}(\theta'_{2}, 3\pi/2-\alpha'_{2},0) &\mbox{\textup{if}}~ \alpha'_{2} \in [0, 3\pi/2], \cr U'_{1}(\theta'_{2}, 7\pi/2-\alpha'_{2},0) &\mbox{\textup{if}}~ \alpha'_{2} \in (3\pi/2, 2\pi),\end{cases}
\end{equation}
\textup{is player 1's best reply to $U'_{2}(\theta'_{2}, \alpha'_{2}, 0)$ as it yields player 1 the payoff $T$. Hence, a possible Nash equilibrium would generate the maximal payoff for player 1. On the other hand, given a fixed player 1's strategy $U'_{1}(\theta'_{1}, \alpha'_{1}, 0)$, player 2 can obtain a payoff that is strictly higher than $S$ by choosing, for example, $U_{2}(\theta'_{2}, \alpha'_{2}, 0)$ with $\theta'_{2} = 0, \alpha'_{2} = 2\pi- \alpha'_{1}$. This means that the player 1 would obtain strictly less than $T$. Hence, there is no pure Nash equilibrium in the game determined by $\Gamma'_{EWL}$. As a result, we can conclude by Lemma~\ref{lemma1} that games~(\ref{p2}) are not strongly isomorphic.}
\end{example}
The example given above shows that the EWL approach with the two-parameter unitary strategies may output different Nash equilibria depending on the order of players' strategies in the classical game. This appears to be a strange feature since games~(\ref{p3}) represent the same decision problem from a game-theoretical point of view. 

One way to make games~(\ref{p2}) isomorphic is to replace player $i$'s strategy set~(\ref{twoparameter}) with the alternative two-parameter strategy space
\begin{equation}
F_{i} = \{U_{i}(\theta_{i}, 0, \beta_{i})\colon \theta_{i} \in [0,\pi], \beta_{i} \in [0,2\pi)\} 
\end{equation}
every time player $i$'s strategies are switched in the classical game. In the case of games~(\ref{p2}) this means that quantum games
\begin{equation}\label{games10}
\Gamma_{EWL} = \left(\{1,2\}, (D_{1}, D_{2}), (M_{i})_{i\in \{1,2\}}\right), \: \Gamma'_{EWL} = \left(\{1,2\}, (D'_{1}, F'_{2}), (M'_{i})_{i\in \{1,2\}} \right) 
\end{equation}
are isomorphic. Indeed, define a game map $\tilde{f} = (\eta, \tilde{\varphi}_{1}, \tilde{\varphi}_{2})$ with $\eta(i) = i$ for $i=1,2$ and bijections $\tilde{\varphi}_{1} \colon D_{1} \to D'_{1}$ and $\tilde{\varphi}_{2}\colon D_{2} \to F'_{2}$ satisfying 
\begin{equation}\label{funkcje10}
\tilde{\varphi}_{1}\left(U_{1}(\theta_{1}, \alpha_{1}, 0)\right) = U'_{1}(\theta_{1}, \alpha_{1}, 0), \: \tilde{\varphi}_{2}\left(U_{2}(\theta_{1}, \alpha_{1}, 0)\right) = U'_{2}(\pi-\theta_{2}, 0, \pi-\alpha_{2}).
\end{equation}
The map $\tilde{\varphi}_{2}$ should actually distinguish cases $\alpha_{2} \in [0,\pi)$ and $\alpha_{2} \in [\pi, 2\pi)$ to be a well-defined bijection as it was done in equation (\ref{equation2cases}). To simplify the proof we stick to the form~(\ref{funkcje10}) throughout the paper bearing in mind that for $\pi - \alpha_{2} \notin [0, 2\pi)$ we can always find the equivalent angle $3\pi - \alpha_{2} \in [0,2\pi)$. We have to show for games~(\ref{games10}) that 
\begin{equation}
u_{i}(U_{1}\otimes U_{2}) = \langle \Psi|M_{i} | \Psi \rangle = \langle \Psi'|M'_{i} | \Psi' \rangle  = u'_{i}(\tilde{f}(U_{1}\otimes U_{2}))
\end{equation}
for $i=1,2$, where $|\Psi\rangle = J^{\dag}(U_{1}\otimes U_{2})J|00\rangle$ and $|\Psi'\rangle = J^{\dag}(\tilde{f}(U_{1}\otimes U_{2}))J|00\rangle$. First, note that $U'_{2}(\pi - \theta_{2}, 0, \pi-\alpha_{2}) = \mathrm{i}\sigma_{x}U'_{2}(\theta_{2}, \alpha_{2}, 0)$. Hence, we obtain
\begin{align}\label{rozwazania10}
|\Psi'\rangle &= J^{\dag}\tilde{f}(U_{1}(\theta_{1}, \alpha_{1}, 0)\otimes U_{2}(\theta_{2}, \alpha_{2}, 0))J|00\rangle\nonumber\\
&= J^{\dag}(U'_{1}(\theta_{1}, \alpha_{1}, 0)\otimes U_{2}(\pi - \theta_{2}, 0, \pi - \alpha_{2}))J|00\rangle\nonumber \\
&=(\mathbb{1}\otimes (-\mathrm{i}\sigma_{x}))J^{\dag}(U'_{1}(\theta_{1}, \alpha_{1}, 0)\otimes U'_{2}(\theta_{2}, \alpha_{2}, 0))J|00\rangle \nonumber\\&= (\mathbb{1}\otimes (-\mathrm{i}\sigma_{x}))|\Psi\rangle.
\end{align}
Application of~(\ref{rozwazania10}) finally yields 
\begin{eqnarray}
\langle \Psi'|M'_{i} | \Psi' \rangle = \langle \Psi | (\mathbb{1}\otimes \mathrm{i} \sigma_{x}) M'_{i} (\mathbb{1}\otimes (-\mathrm{i} \sigma_{x}))|\Psi\rangle = \langle \Psi |M_{1}|\Psi \rangle.
\end{eqnarray}
In similar way we can prove a more general fact. Namely, if $F_{2}$ is player 2's strategy set in one of games~(\ref{p2}) and $D_{2}$ is in the the other one then games~(\ref{p2}) become strongly isomorphic. This observation suggests that the EWL scheme is robust with respect to changing the order of players' strategies in the classical game if the players can use strategies from $D_{i} \cup F_{i}$, or equivalently from the set $\mathsf{SU}(2)$. Before stating the general result we study a specific example.
\begin{example}\label{exampleewlfull}
\textup{Let us consider the following three-person games:}
\begin{eqnarray}
v \quad \bordermatrix{& l & r \cr t & (a_{000}, b_{000}, c_{000}) & (a_{010}, b_{010}, c_{010}) \cr b & (a_{100}, b_{100}, c_{100}) & (a_{110}, b_{110}, c_{110})} \qquad w \quad \bordermatrix{& l & r \cr t & (a_{001}, b_{001}, c_{001}) & (a_{011}, b_{011}, c_{011}) \cr b & (a_{101}, b_{101}, c_{101}) & (a_{111}, b_{111}, c_{111})} \nonumber
\end{eqnarray}
\textup{and}
\begin{eqnarray}
v' \quad \bordermatrix{& l' & r' \cr t' & (a_{000}, b_{000}, c_{000}) & (a_{010}, b_{010}, c_{010}) \cr b' & (a_{100}, b_{100}, c_{100}) & (a_{110}, b_{110}, c_{110})} \qquad w' \quad \bordermatrix{& l' & r' \cr t' & (a_{001}, b_{001}, c_{001}) & (a_{011}, b_{011}, c_{011}) \cr b' & (a_{101}, b_{101}, c_{101}) & (a_{111}, b_{111}, c_{111})}.\nonumber
\end{eqnarray}
\textup{The games are (strongly) isomorphic via game mapping $f=(\eta, \varphi_{1}, \varphi_{2}, \varphi_{3})$ such that}
\begin{align}\label{iso1}\begin{split}
&\eta = (1 \to 2, 2 \to 3, 3 \to 1), \cr &\varphi_{1} = (t \to l', b\to r'),\, \varphi_{2} = (l \to w', r \to v'),\, \varphi_{3} = (v \to b', w \to t'). \end{split}
\end{align}
\textup{We see from (\ref{iso1}) that the isomorphism maps strategy profiles as follows:}
\begin{align}
\begin{split}
&f(t,l,v) = (b',l',w'),\, f(t,r,v) = (b',l',v'),\, f(b,l,v) = (b',r',w'), \cr 
&f(b,r,v) = (b',r',v'),\, f(t,l,w)= (t',l',w'),\, f(t,r,w) = (t',l',v') \cr &f(b,l,w) = (t',r',w'),\, f(b,r,w) = (t',r',v'). \end{split}
\end{align}

\textup{Let us now define the EWL quantum extensions $\Gamma_{EWL}$ and $\Gamma'_{EWL}$ for the three-player game where we identify the players' first and second strategies with values 0 and 1, respectively. That is,}
\begin{equation}\label{2ewl}
\Gamma_{EWL} = (N, (D_{i})_{i\in N}, (M_{i})_{i\in N})\quad\mbox{\textup{and}}\quad\Gamma'_{EWL} = (N, (D'_{i})_{i\in N}, (M'_{i})_{i\in N})
\end{equation}
\textup{where} $N = \{1,2,3\}$, $D_{i} = D'_{i} = \mathsf{SU}(2)$ \textup{for each} $i\in N$,
\begin{align}
&(M_{1}, M_{2}, M_{3}) = \sum_{j_{1}, j_{2}, j_{3}=0,1}(a_{j_{1}j_{2}j_{3}},b_{j_{1}j_{2}j_{3}}, c_{j_{1}j_{2}j_{3}})P_{j_{1}j_{2}j_{3}},
\\
&(M'_{1}, M'_{2}, M'_{3}) = \sum_{j_{1}, j_{2}, j_{3}=0,1}(c_{j_{1}j_{2}j_{3}}, a_{j_{1}j_{2}j_{3}}, b_{j_{1}j_{2}j_{3}})P_{f({j_{1}j_{2}j_{3}})},
\end{align}
\textup{where $P_{j_{1}j_{2}j_{3}} = |j_{1}j_{2}j_{3}\rangle \langle j_{1}j_{2}j_{3}|$. Given $f = (\eta, (\varphi_{i})_{i\in N})$ let us define a mapping $\tilde{f} = (\eta, (\tilde{\varphi_{i}})_{i\in N})$ such that $\tilde{\varphi_{i}}\colon D_{i} \to D'_{\eta(i)}$ for $i\in N$ and}
\begin{align}\begin{split}
&\tilde{\varphi_{1}}(U_{1}(\theta_{1}, \alpha_{1}, \beta_{1})) = U'_{2}(\theta_{1}, \alpha_{1}, \beta_{1}), \cr &\tilde{\varphi_{2}}(U_{2}(\theta_{2}, \alpha_{2}, \beta_{2})) = U'_{3}(\pi - \theta_{2}, 2\pi - \beta_{2}, \pi - \alpha_{2}), \cr &\tilde{\varphi_{3}}(U_{3}(\theta_{3}, \alpha_{3}, \beta_{3})) = U'_{1}(\pi - \theta_{3}, 2\pi - \beta_{3}, \pi - \alpha_{3}).
\end{split}
\end{align}
\textup{Then, $\tilde{f}$ induces a bijection from $D_{1}\otimes D_{2} \otimes D_{3}$ to $D'_{1}\otimes D'_{2} \otimes D'_{3}$ such that}
\begin{align}
\tilde{f}(U_{1}\otimes U_{2} \otimes U_{3}) &= (\tilde{\varphi_{3}}(U_{3}(\theta_{3}, \alpha_{3}, \beta_{3})), \tilde{\varphi_{1}}(U_{1}(\theta_{1}, \alpha_{1}, \beta_{1}))\otimes \tilde{\varphi}_{2}(U_{2}(\theta_{2}, \alpha_{2}, \beta_{2}))) \nonumber \\ &=U'_{1}(\pi-\theta_{3}, 2\pi - \beta_{3}, \pi-\alpha_{3})\otimes U'_{2}(\theta_{1}, \alpha_{1}, \beta_{1})\otimes U'_{3}(\pi - \theta_{2}, 2\pi - \beta_{2}, \pi - \alpha_{2}). \nonumber 
\end{align}
\textup{According to the EWL scheme, the payoff functions for $\Gamma_{EWL}$ and $\Gamma'_{EWL}$ are as follows:}
\begin{align}\begin{split}
&u_{i}(U_{1}\otimes U_{2} \otimes U_{3}) = \langle \Psi|M_{i}|\Psi\rangle,\,\mbox{\textup{where}}\,|\Psi\rangle = J^{\dag}(U_{1}\otimes U_{2} \otimes U_{3})J|000\rangle \cr &u'_{i}(U'_{1}\otimes U'_{2} \otimes U'_{3}) = \langle \Psi'|M'_{i}|\Psi'\rangle,\,\mbox{\textup{where}}\,|\Psi'\rangle = J^{\dag}(U'_{1}\otimes U'_{2} \otimes U'_{3})J|000\rangle \end{split}
\end{align}
\textup{for $i\in N$. In order to prove that $\Gamma_{EWL}$ and $\Gamma'_{EWL}$ are isomorphic we have to check if}
\begin{equation}
u_{i}(U_{1}\otimes U_{2} \otimes U_{3}) = u'_{\eta(i)}(\tilde{f}(U_{1}\otimes U_{2} \otimes U_{3}))\; \mbox{\textup{for}}\; i\in N. 
\end{equation}
\textup{Without loss of generality we can assume that $i=1$. Let us first evaluate state $|\Psi'\rangle$,}
\begin{equation}\label{psi}
|\Psi'\rangle = J^{\dag}\left(U'_{1}(\pi-\theta_{3}, 2\pi - \beta_{3}, \pi-\alpha_{3})\otimes U'_{2}(\theta_{1}, \alpha_{1}, \beta_{1})\otimes U'_{3}(\pi - \theta_{2}, 2\pi - \beta_{2}, \pi - \alpha_{2})\right)J|000\rangle.
\end{equation}
\textup{Note that}
\begin{align}\label{redukcja}
&U'_{1}(\pi - \theta_{3}, 2\pi - \beta_{3}, \pi-\alpha_{3}) \otimes U'_{2}(\theta_{1}, \alpha_{1}, \beta_{1}) \otimes U'_{3}(\pi - \theta_{2}, 2\pi - \beta_{2}, \pi - \alpha_{2}) \nonumber\\
&\quad= (-\sigma_{x}\otimes \mathbb{1} \otimes \sigma_{x})(U'_{1}(\theta_{3},\alpha_{3}, \beta_{3}) \otimes U'_{2}(\theta_{1}, \alpha_{1}, \beta_{1}) \otimes U'_{3}(\theta_{2}, \beta_{2}, \alpha_{2}))
\end{align}
\textup{and}
\begin{align}\label{Spi}
&U'_{1}(\theta_{3}, \alpha_{3}, \beta_{3})\otimes U'_{2}(\theta_{1}, \alpha_{1}, \beta_{1})\otimes U'_{3}(\theta_{2}, \alpha_{2}, \beta_{2})\nonumber\\ &\quad=S_{\eta}\left(U'_{2}(\theta_{1}, \alpha_{1}, \beta_{1}) \otimes U'_{3}(\theta_{2}, \alpha_{2}, \beta_{2})\otimes U'_{1}(\theta_{3}, \alpha_{3}, \beta_{3})\right)S^{\dag}_{\eta},
\end{align}
\textup{where $S_{\eta}$ is a permutation matrix that changes the order of qubits according to $\eta$,}
\begin{align}
S_{\eta} &= |000\rangle \langle 000| + |001\rangle \langle 010| + |010\rangle \langle 100| + |011\rangle \langle 110| \nonumber \\ &\quad + |100\rangle \langle 001| + |101\rangle \langle 011| + |110\rangle \langle 101| + |111\rangle \langle 111|.
\end{align}
\textup{Using (\ref{redukcja}), (\ref{Spi}), the fact that $[J^{\dag}, -\sigma_{x}\otimes \mathbb{1}\otimes \sigma_{x}] =[J^{\dag}, S_{\eta}] = [J, S_{\eta}] = 0$ and $S_{\eta}^{\dag}|000\rangle = |000\rangle$ we may write $|\Psi'\rangle$ as follows:}
\begin{align}\label{psiprim}
|\Psi'\rangle &= -\left(\sigma_{x}\otimes \mathbb{1} \otimes \sigma_{x}\right)S_{\eta}J^{\dag}\left(U'_{2}(\theta_{1}, \alpha_{1}, \beta_{1})\otimes U'_{3}(\theta_{2}, \alpha_{2}, \beta_{2}) \otimes U'_{1}(\theta_{3}, \alpha_{3}, \beta_{3})\right)J|000\rangle \nonumber \\ &=-\left(\sigma_{x}\otimes \mathbb{1} \otimes \sigma_{x}\right)S_{\eta}|\Psi\rangle
\end{align}
\textup{Note that $\langle j_{1}j_{2}j_{3}|S_{\eta} = \left(S^{\dag}_{\eta}|j_{1}j_{2}j_{3}\rangle\right)^{\dag}$. This means that $S_{\eta}$ is the inverse operation when acting on dual vectors. This observation together with the fact that $f$ changes the strategy order for player 1 and 3 lead us to conclusion that operator $\pm(\sigma_{x}\otimes \mathbb{1} \otimes \sigma_{x})S_{\eta}$ can be viewed as $f^{-1}$ in the sense of the following equality:}
\begin{equation}\label{endeq}
|\langle j_{1}j_{2}j_{3}|(-\sigma_{x}\otimes \mathbb{1} \otimes \sigma_{x})S_{\eta}| = |\langle f^{-1}(j_{1}j_{2}j_{3})|.
\end{equation}
\textup{Let us now consider term $\langle \Psi'|P_{f(j_{1}j_{2}j_{3})}|\Psi'\rangle$ for $|\Psi'\rangle$ given by (\ref{psi}). From~(\ref{psiprim}) and (\ref{endeq}) it follows that}
\begin{align}
\langle \Psi'|P_{f(j_{1}j_{2}j_{3})}|\Psi'\rangle &= |\langle f(j_{1}j_{2}j_{3})|\Psi'\rangle|^2 = |\langle f(j_{1}j_{2}j_{3})|(\sigma_{x}\otimes \mathbb{1} \otimes \sigma_{x})S_{\eta}|\Psi\rangle|^2 \nonumber\\&= |\langle j_{1}j_{2}j_{3}|\Psi\rangle|^2 = \langle\Psi|P_{j_{1}j_{2}j_{3}}|\Psi\rangle.
\end{align}
\textup{Hence,}
\begin{equation}\label{finaleq}
u_{\eta(1)}(\tilde{f}(U_{1}\otimes U_{2} \otimes U_{3})) = \langle \Psi'|M'_{\eta(1)}|\Psi'\rangle = \langle \Psi|M_{1}|\Psi\rangle = u_{1}(U_{1}\otimes U_{2} \otimes U_{3}).
\end{equation}
\textup{Similar reasoning applies to the case $i=2,3$. We have thus proved that games given by (\ref{2ewl}) are isomorphic.}
\end{example}
The same conclusion can be drawn for games with arbitrary but finite number $N$ of players.
\begin{proposition}
Let $\Gamma = (N, (S_{i})_{i\in N}, (u_{i})_{i\in N})$ and $\Gamma' = (N, (S'_{i})_{i\in N}, (u'_{i})_{i\in N})$ be strongly isomorphic strategic form games with $|S_{i}| = |S'_{i}| = 2$ and let $\Gamma_{EWL} = (N, (D_{i})_{i\in N}, (M_{i})_{i\in N})$ and $\Gamma'_{EWL} = (N, (D'_{i})_{i\in N}, (M'_{i})_{i\in N})$ with $D_{i} = D'_{i} = \mathsf{SU}(2)$ be the corresponding quantum games. Then $\Gamma_{EWL}$ and $\Gamma'_{EWL}$ are strongly isomorphic.
\end{proposition}
\begin{proof}
The proof follows by the same method as in Example~\ref{exampleewlfull}. Let $f=(\eta, (\varphi_{i})_{i\in N})$ be a strong isomorphism between $\Gamma$ and $\Gamma'$. Depending on $\varphi_{i}\colon S_{i} \to S'_{\eta(i)}$ such that $\varphi_{i}(s^i_{k}) = s^{\eta(i)}_{l}$ for $A_{i} = \{s^{i}_{0},s^{i}_{1}\}$ and $A_{\eta(i)} = \{s^{\eta(i)}_{0},s^{\eta(i)}_{1}\}$ we construct $\tilde{f} = (\eta, (\tilde{\varphi_{i}})_{i\in N})$ where
\begin{equation}\label{podwojnywzor}
\tilde{\varphi_{i}}(U_{i}(\theta_{i}, \alpha_{i}, \beta_{i})) = \begin{cases}U'_{\eta(i)}(\theta_{i}, \alpha_{i}, \beta_{i}) &\mbox{if}~ \varphi_{i}\left(s^i_{k}\right) = s^{\eta(i)}_{k} \cr U'_{\eta(i)}(\pi -\theta_{i},2\pi - \beta_{i}, \pi -\alpha_{i}) & \mbox{if}~ \varphi_{i}\left(s^i_{k}\right) = s^{\eta(i)}_{k\oplus_{2}1}.\end{cases}
\end{equation}
Then $\tilde{f}\left(\bigotimes^N_{i=1}U_{i}\right) = \bigotimes^N_{i=1}U'_{i}$, where $U'_{\eta(i)} = \tilde{\varphi_{i}}(U_{i})$ for $i=1,\dots N$. Since $\eta$ is a permutation and $U(\pi - \theta, 2\pi - \beta, \pi-\alpha) = -\mathrm{i}\sigma_{x}U(\theta, \alpha, \beta)$, we can write relation~(\ref{podwojnywzor}) as
\begin{equation}\label{podwojnywzor2}
\tilde{\varphi}_{\eta^{-1}(i)}(U_{\eta^{-1}(i)}(\theta_{\eta^{-1}(i)}, \alpha_{\eta^{-1}(i)}, \beta_{\eta^{-1}(i)})) = \begin{cases}U'_{i}(\theta_{\eta^{-1}(i)}, \alpha_{\eta^{-1}(i)}, \beta_{\eta^{-1}(i)}) &\mbox{if}~ \varphi_{i}\left(s^i_{k}\right) = s^{\eta(i)}_{k} \cr -\mathrm{i}\sigma_{x}U'_{i}(\theta_{\eta^{-1}(i)},\alpha_{\eta^{-1}(i)}, \beta_{\eta^{-1}(i)}) & \mbox{if}~ \varphi_{i}\left(s^i_{k}\right) = s^{\eta(i)}_{k\oplus_{2}1}.\end{cases}
\end{equation}
As a result, $\tilde{f}$ maps $\bigotimes^{N}_{i=1}U_{\eta^{-1}(i)}$ onto $\bigotimes^{N}_{i=1}U'_{i}$ as follows:
\begin{equation}
\tilde{f}\left(\bigotimes^{N}_{i=1}U_{i}\right) = \bigotimes^{N}_{i=1}V_{i}\bigotimes^N_{i=1}U'_{i}(\theta_{\eta^{-1}(i)}, \alpha_{\eta^{-1}(i)}, \beta_{\eta^{-1}(i)}), \quad V_{i} = \begin{cases}\mathbb{1} &\mbox{if}~ \varphi_{i}\left(s^i_{k}\right) = s^{\eta(i)}_{k} \cr -\mathrm{i}\sigma_{x} &\mbox{if}~ \varphi_{i}\left(s^i_{k}\right) = s^{\eta(i)}_{k\oplus_{2}1}.\end{cases}
\end{equation}
Let us now consider a permutation matrix $S_{\eta} \in M_{2^{N}}$ that rearranges the order of basis states $\{|j_{i}\rangle\} \in \{|0\rangle, |1\rangle\}$  in the tensor product $|j_{1}\rangle|j_{2}\rangle\dots|j_{N}\rangle$. Since $S_{\eta}$ permutes the elements in a similar way as $\tilde{f}$, it is not difficult to see that
\begin{equation}
S_{\eta}\bigotimes^{N}_{i=1}U_{i}(\theta_{i}, \alpha_{i}, \beta_{i})S^{T}_{\eta} = \bigotimes^{N}_{i=1}U_{i}(\theta_{\eta^{-1}(i)}, \alpha_{\eta^{-1}(i)}, \beta_{\eta^{-1}(i)})
\end{equation}
It is also clear that $\sigma_{x}^{\otimes N}$ commutes with $\bigotimes^{N}_{i=1}V_{i}$ and $S_{\eta}$ and so does $J = (\mathbb{1}^{\otimes N} + \mathrm{i}\sigma_{x}^{\otimes N})/\sqrt{2}$. Thus the final state $|\Psi'\rangle = J^{\dag}\tilde{f}\left(\bigotimes^{N}_{i=1}U_{i}(\theta_{i}, \alpha_{i}, \beta_{i})\right)J|0\rangle^{\otimes N}$ may be written as
\begin{equation}
|\Psi'\rangle = \bigotimes^N_{i=1}V_{i}S_{\eta}J^{\dag}\left(\bigotimes^{N}_{i=1}U_{i}(\theta_{i}, \alpha_{i}, \beta_{i})\right)J|0\rangle^{\otimes N} = \bigotimes^N_{i=1}V_{i}S_{\eta}|\Psi\rangle.
\end{equation}
Analysis similar to that in equations~(\ref{endeq})-(\ref{finaleq}) shows that 
\begin{equation}
u_{\eta(i)}\left(\tilde{f}\left(\bigotimes^{N}_{i=1}U_{i}\right)\right) = u_{i}\left(\bigotimes^{N}_{i=1}U_{i}\right),
\end{equation}
which is the desired conclusion.
\end{proof}
As the following example shows, the converse is not true in general.
\begin{example}
\textup{Let us consider two $2\times 2$ bimatrix games that differ only in the order of payoff profiles in the anti-diagonal, i.e.,}
\begin{equation}\label{antidiagonal}
\Gamma\colon \quad \bordermatrix{& l & r \cr t & (a_{00}, b_{00}) & (a_{01}, b_{01}) \cr b & (a_{10}, b_{10}) & (a_{11}, b_{11})} \quad \mbox{\textup{and}} \quad \Gamma'\colon \quad \bordermatrix{& l' & r' \cr t' & (a_{00}, b_{00}) & (a_{10}, b_{10}) \cr b' & (a_{01}, b_{01}) & (a_{11}, b_{11})}.
\end{equation}
\textup{The EWL quantum counterparts $\Gamma_{EWL}$ and $\Gamma'_{EWL}$ for these games are specified by triples (\ref{2ewl}), where in this case $N = \{1,2\}$, $D_{i} = D'_{i} = \mathsf{SU}(2)$ and the measurement operators take the form}
\begin{equation}
(M_{1}, M_{2}) = \sum_{j_{1},j_{2} = 0,1}(a_{j_{1}j_{2}}, b_{j_{1}j_{2}})P_{j_{1}j_{2}}, \quad (M'_{1}, M'_{2}) = \sum_{j_{1},j_{2} = 0,1}(a_{j_{1}j_{2}}, b_{j_{1}j_{2}})P_{j_{2}j_{1}},
\end{equation}
\textup{where $P_{j_{1}j_{2}} = |j_{1}j_{2}\rangle \langle j_{1}j_{2}|$. Let us set a mapping $\tilde{f} = (\eta, (\tilde{\varphi}_{1}, \tilde{\varphi}_{2}))$ with $\eta(i) = i$ and $\tilde{\varphi}_{i}(U_{i}(\theta_{i}, \alpha_{i}, \beta_{i}) = U'_{i}(\pi - \theta_{i}, \pi/4-\beta_{i}, \pi/4-\alpha_{i})$) for $i=1,2$. An easy computation shows that}
\begin{align}\label{sfj}
|\Psi'\rangle &= J^{\dag}\tilde{f}(U_{1}\otimes U_{2})J|00\rangle\nonumber\\ &= J^{\dag}\left(U_{1}\left(\pi-\theta_{1}, \frac{\pi}{4}-\beta_{1}, \frac{\pi}{4}-\alpha_{1}\right)\otimes U_{2}\left(\pi-\theta_{2}, \frac{\pi}{4}-\beta_{2}, \frac{\pi}{4} - \alpha_{2}\right)\right)J|00\rangle \nonumber \\ &=SFJ^{\dag}(U_{1}(\theta_{1}, \alpha_{1}, \beta_{1})\otimes U_{2}(\theta_{2}, \alpha_{2}, \beta_{2}))J|00\rangle = SF|\Psi\rangle, 
\end{align}
\textup{where $S$ has the outer product representation $S = |00\rangle \langle 00| + |01\rangle \langle 10| + |10\rangle \langle 01| + |11\rangle \langle 11|$ and $F = |00\rangle \langle 00| + |01\rangle \langle 10| + |10\rangle \langle 01| - |11\rangle \langle 11|$. Application of equation~(\ref{sfj}) gives}
\begin{equation}
u'_{i}\left(\tilde{f}(U_{1}\otimes U_{2})\right) = \langle \Psi'|M'_{i}|\Psi'\rangle = \langle \Psi|FSM'_{i}SF|\Psi\rangle = \langle \Psi|M_{i}|\Psi\rangle = u_{i}(U_{1}\otimes U_{2}).
\end{equation}
\textup{As a result, games produced by $\Gamma_{EWL}$ and $\Gamma'_{EWL}$ are strongly isomorphic. This fact, however, is not sufficient to guarantee the isomorphism between $\Gamma$ and $\Gamma'$. Indeed, one can check that there is no $f=(\eta, (\varphi_{1}, \varphi_{2}))$ to satisfy $u_{i}(s) = u'_{\eta(i)}(f(s))$ for each $s \in \{t,b\}\times \{l,r\}$ and $i=1,2$. Alternatively, given specific payoff profiles $(a_{00},b_{00}) = (4,4), (a_{01}, b_{01}) = (1,3), (a_{10}, b_{10}) = (3,1), (a_{11}, b_{11}) = (2,2)$, we can find three Nash equilibria in the game $\Gamma$ and just one in the game $\Gamma'$. Hence, by Lemma~\ref{lemma1} games (\ref{antidiagonal}) are not isomorphic. }
\end{example}
\section{Conclusions}
The theory of quantum games has no rigorous mathematical structure. There are no formal axioms, definitions that would give clear directions of how a quantum game ought to look like. In fact, only one condition is taken into consideration. It says that a quantum game ought to include the classical way of playing the game. As a result, this allows us to define a quantum game scheme in many different ways. The scheme we have studied in the paper is definitely ingenious. It has made a significant contribution to quantum game theory. However, it leaves the freedom of choice of the players' strategy sets. Our criterion for quantum strategic game schemes requires the quantum model to preserve strong isomorphism. This specifies the strategy sets to be $\mathsf{SU}(2)$. We have shown that a proper subset of $\mathsf{SU}(2)$ in the EWL scheme may imply different quantum counterparts of the same game-theoretical problem. In that case, the resulting quantum game (in particular, its Nash equilibria) depends on the order of players' strategies in the input bimatrix game. Hence, given a classical game, for example the Prisoner's Dilemma, we cannot say anything about the properties of the EWL approach with the two-parameter unitary strategies until we specify an explicit bimatrix for that game. This is not the case in the EWL scheme with $\mathsf{SU}(2)$ where, given a classical bimatrix game or its isomorphic counterpart, we always obtain the same from the game-theoretical point of view quantum game.
\section*{Acknowledges}
We would like to thank the Reviewers for discussion which undoubtedly improved the quality of our work.

\noindent This work was supported by the Ministry of Science and Higher Education in Poland under the project Iuventus Plus IP2014 010973 in the years 2015–2017.





\begin{thebibliography}{100}
\bibitem{mejer}  Meyer D A 1999 Quantum Strategies {\it Phys. Rev. Lett.} {\bf 82} 1052–55
\bibitem{ewl} Eisert J Wilkens M and Lewenstein M 1999 Quantum Games and Quantum Strategies
{\it Phys. Rev. Lett.} {\bf 83} 3077-80
\bibitem{marinatto} Marinatto L and Weber T 2000 A quantum approach to static games of complete information {\it Phys. Lett.} A {\bf 272} 291
\bibitem{du} Li H Du J and Massar S 2002 Continuous-variable quantum games {\it Phys. Lett.} A {\bf 306} 73–8
\bibitem{fracormwcomment} Frackiewicz P 2013 A comment on the generalization of the Marinatto-Weber quantum game scheme, Acta. Phys. Pol. B  {\bf 44} 29
\bibitem{bleiler} Bleiler S A 2008 A formalism for quantum games and an application, preprint arxiv:0808.1389v1 [quant-ph] available at http://arxiv.org/abs/0808.1389
\bibitem{experimental} Du J Li H Xu X Shi M Wu J Zhou X and Han R 2002 Experimental realization of quantum games on a quantum computer, {\it Phys. Rev. Lett} {\bf 88} 137902
\bibitem{chen} Chen L K Ang H Kiang D Kwek L C and Lo C F 2003 Quantum prisoner dilemma under decoherence {Phys. Lett.} A {\bf 316} 317
\bibitem{liiqbal} Li Q Iqbal A Chen M and Abbot D 2012 Quantum strategies win in a defector-dominated population {\it Physica A} {\bf 391} 3316
\bibitem{nawaznowy} Nawaz A 2013 The strategic form of quantum Prisoners' Dilemma {\it Chinese Phys. Lett.} {\bf 30} 050302
\bibitem{benjamincomment} Benjamin S C and Hayden P M 2001 Comment on ''Quantum games and quantum strategies'' {\it Phys. Rev. Lett.} {\bf 87} 069801
\bibitem{flitney} Flitney A P and Hollenberg L C L 2007 Nash equilibria in quantum games with generalized two-parameter strategies {\it Phys. Lett. A} {\bf 363} 381
\bibitem{fracor2015} Fr\c{a}ckiewicz P 2015 A new quantum scheme for normal form games, {\it Quantum Inf. Process.} {\bf 14} 1809
\bibitem{maschler} Maschler M Solan E and Zamir S 2013 Game Theory, Cambridge University Press.
\bibitem{garcia} Gabarr\'o J Garc\'ia A and Serna M 2011 The complexity of game isomorphism {\it Theor. Comput. Sci.} {\bf 412} 6675
\bibitem{nash} Nash J 1951 Non-cooperative games {\it Ann Math} {\bf 54} 286
\bibitem{peleg1} Peleg B Rosenm\"uller J and Sudh\"olter P 1999 The Canonical Extensive Form of a Game Form - Part I - Symmetries {\it In Current Trends in Economics, Advancement of Studies in Economics} 367
\bibitem{peleg2} Sudh\"olter P Rosenm\"uller J and Peleg B 2000 The Canonical Extensive Form of a Game Form - Part II - Representation {\it J Math Econ} {\bf 33} 299
\end{thebibliography}
\end{document}